\UseRawInputEncoding

\documentclass[a4paper,11pt,reqno]{amsart}

\usepackage{amssymb, amsmath, bm}
\usepackage{mathrsfs}
\usepackage{dsfont}
\usepackage{appendix}

\allowdisplaybreaks

\newtheorem{thm}{Theorem}[section]
\newtheorem{prp}[thm]{Proposition}
\newtheorem{lem}[thm]{Lemma}
\newtheorem{dfn}[thm]{Definition}

\newtheorem{cor}[thm]{Corollary}
\newtheorem{example}[thm]{Example}

\newtheorem{remark}[thm]{Remark}

\newenvironment{prf}{\noindent {\it Proof.}  }{\hfill $\Box$}

\newcommand{\pa}{\partial}

\newcommand\Z{\mathbb{Z}}

\newcommand{\res}{\mathrm{Res}}

\newcommand{\bs}[1]{\boldsymbol{#1}}

\begin{document}

\title[]{The Virasoro symmetries of the bigraded modified Toda hierarchy}
\author[]{Yi Yang$^*$}
	\thanks{*Corresponding author.}
\dedicatory {School of Mathematics, Sun Yat-sen University\\
Guangzhou, 510000, China  \\
Email address:  yangy875@mail2.sysu.edu.cn}
\begin{abstract}
Modified Toda hierarchy is  a two-component generalization of the 1st modified KP hierarchy, which has been widely applied to analyze constraints of the Toda hierarchy, including the B--Toda and C--Toda hierarchies.  In this paper, we construct additional symmetries for the modified Toda  hierarchy and derive the corresponding Adler-Shiota-van Moerbeke formula. In addition, we also show that $(N,M)$--bigraded modified Toda hierarchy(BMTH), which is a special reduction for modified Toda, possess additional symmetries that form  a subalgebra of the Virasoro algebra.
\end{abstract}
\maketitle

\section{Introduction}
Modified Toda, also  known as the two--component 1st modified Kadomtsev--Petviashvili (KP) hierarchy, was first    introduced in \cite{van2015,Jimbo1983} as a bilinear equation satisfied by tau functions. It has been recently reformulated in terms of a Lax pair $(L_1,L_2)$ involving difference operators \cite{mtoda}.  This hierarchy has been shown to be closely related
 to  various integrable systems. As detailed in \cite{mtoda}, there exist the Miura links between modified Toda and two--dimensional Toda lattice (Toda for short) hierarchy, which is one of the most important research objects in the mathematical physics and integrable systems.  Besides this,
 some different types of constraints of Toda are sub--hierarchies of
modified Toda hierarchy, including  constrained Toda(C--Toda)\cite{Krichever2022} and Toda lattice with the constraint of type B (B--Toda)\cite{Prokofev2023,Krichever2023}.  At present,  modified Toda has been studied from various aspects, including Darboux transformation\cite{mtoda}, reduction and extensions\cite{BMTH}, among others. The main goal of this paper  is to construct additional symmetries for modified Toda hierarchy.

One of the attractive topics on hierarchies of integrable systems is to study their
additional symmetries, which can be used to derive string equations and Virasoro constraints in the matrix models of the string theory\cite{Dick-string}.  Such symmetries can be represented by the Lax equations constructed by Orlov and Schulman\cite{OS}, which involve operators explicitly dependent on space and time variables. In addition to this, there exists the Sato B\"acklund symmetry defined by the vertex operator acting on tau function\cite{DKJM}.  These two types of symmetries in fact can be connected
 through the so-called Adler-Shiota-van Moerbeke(ASvM) formula\cite{ASVM} in consideration of KP and the Toda  hierarchies.  Almost at the same time, the proof of ASvM formula of KP was presented by Dickey\cite{add-KP} using the notion of resolvent operators and the Fay identity of the KP hierarchy. In \cite{van-BKP},
 van de Leur also proposed a novel approach to obtain the corresponding ASvM formula for BKP hierarchy by using an algebraic formalism. These methods have subsequently been generalized to other integrable systems, such as BKP \cite{Tu-BKP}, CKP \cite{HTFM-BKP}, modified KP \cite{CLT-mKP}, two-component BKP \cite{Wu-2BKP}, and so on. Building on the Toda hierarchy framework, we attempt to derive additional symmetries, and give an alternative
proof of the ASvM formula for the modified Toda hierarchy based on the approach in \cite{ASVM}.

A special type of  reduction of modified Toda  is the  $(N,M)$--bigraded modified  Toda hierarchy (BMTH) labelled by positive integers $N$ and $M$, which  is  associated to affine algebra ${\rm\widetilde{sl}}_{n}$\cite{BMTH}.  This hierarchy can be derived by imposing the constrains $L_1^N=L_2^M$ and  $ (L_1^N+L_2^M)(1)=0$. Such constrains cannot make it possible for the additional symmetries of the modified Toda to be directly reduced to those of the $(N,M)$--BMTH.  However, a reduction is feasible for certain linear combinations of these symmetries. It can be shown in this paper that additional symmetries of the $(N,M)$--BMTH form a subalgebra of the Virasoro algebra with zero central charge.

The paper is structured as follows: In Section 2, some basic facts about the modified Toda hierarchy and difference operators are introduced. Then in Section 3, the additional symmetry and ASvM formula for modified Toda are discussed. After that
in Section 4,
the Virasoro symmetries of the $(N,M)$--BMTH are given.   Section 5 is devoted to conclusions and discussions.

\section{Modified Toda hierarchy and difference operators}
Let us first give a brief introduction to the modified Toda   hierarchy based on detailed research in \cite{mtoda}.
Assume $\mathcal{F}$ is a commutative associative algebra of smooth complex functions in  variable $s$.
Denote by $\mathcal{F}[[\Lambda,\Lambda^{-1}]]$ the space of  formal pseudo difference operators, which consists of all expressions of the form $$A=\sum\limits_{m\in \mathbb{Z}}a_m(s)\Lambda^m,\quad a_m(s)\in\mathcal{F},$$
where $\Lambda$ is the shift operator acting on function by
$
\Lambda (f(s))=f(s+1).$  We can define the adjoint operation $*$ by $
A^*=\sum_{m\in \mathbb{Z}}\Lambda^{-m}a_m(s),$
which obeys $(AB)^*=B^*A^*$, $f^*=f$ and $(Af)^*=fA^*$ for $f\in \mathcal{F}$ and $A,B\in \mathcal{F}[[\Lambda,\Lambda^{-1}]]$ .

Let $\mathcal{F}((\Lambda))$ and $\mathcal{F}((\Lambda^{-1}))$ be two subspaces of $\mathcal{F}[[\Lambda,\Lambda^{-1}]]$, defined by requiring the coefficients $a_i$
to vanish for $i\ll0$ and $i\gg0$ in $A$ respectively. Both $\mathcal{F}((\Lambda))$ and $\mathcal{F}((\Lambda^{-1}))$ are the associative rings, where  the multiplication is defined by
 $
(f(s)\Lambda^i)(g(s)\Lambda^j)=f(s)g(s+i)\Lambda^{i+j},
$
Introduce  operator $\Delta =\Lambda-1$ and its adjoint operator $\Delta^* =\Lambda^{-1}-1$, and define the following multiplication operations,  for any
$j\in \mathbb{Z}$
\begin{align*}
&\Delta^j\cdot f=\sum_{i=0}^\infty\binom{j}{i}(\Delta^i(f))(n+j-i)\Delta^{j-i},\quad \Delta^{*j}\cdot f=\sum_{i=0}^\infty\binom{j}{i}(\Delta^{*i}(f))(n+j-i)\Delta^{*j-i}.
\end{align*}
Hence we obtain two associative rings $\mathcal{F}((\Delta^{-1}))$ and $\mathcal{F}((\Delta^{*-1}))$.

In what follows, we define the following expansions
\begin{align}
(\Delta+1)^{-1}=\Delta^{-1}\sum_{k=0}^{\infty}(-\Delta^{-1})^k,&\quad (\Delta^*+1)^{-1}=\Delta^{*-1}\sum_{k=0}^{\infty}(-\Delta^{*-1})^{k},\label{delta-lambda}\\
(\Lambda-1)^{-1}=\sum_{k=1}^{\infty}\Lambda^{-k},&\quad
(\Lambda^{-1}-1)^{-1}=\sum_{k=1}^{\infty}\Lambda^k,\label{lambda- delta}
\end{align}
then it follows that $\mathcal{F}((\Lambda^{-1}))=\mathcal{F}((\Delta^{-1}))$  and  $\mathcal{F}((\Lambda))=\mathcal{F}((\Delta^{*-1})). $
\begin{lem}\label{lemma:relation1}\cite{mtoda}
For any  formal difference operators $A\in \mathcal{F}[[\Lambda,\Lambda^{-1}]]$, we have
\begin{align}
A_{\Delta,\geq0}=A_{\Lambda,\geq0},&\quad A_{\Delta,[0]}=A_{\Lambda,\geq0}(1),\label{ALambda>0}\\
A_{\Delta^*,\geq0}=A_{\Lambda,\leq0},&\quad A_{\Lambda,\leq0}(1)=A_{\Delta^*,[0]},\label{ALambda<0}\\
(A\Lambda\Delta^{-1})_{\Lambda,\geq1}\Delta\Lambda^{-1}=A_{\Delta,\geq1},&\quad(A\Delta^{*-1})_{\Lambda,\leq0}\Delta^{*}=A_{\Delta^*,\geq1},\label{ALambdadelta}
\end{align}
where $(\quad)_{p,I}$ denote $(A)_{p,I}=\sum\limits_{m\in I}a_mp^m$ with $I$ is the subset of $\mathbb{Z}$ and $p\in\{\Lambda,\Delta,\Delta^*\}$.
\end{lem}
\begin{lem}\label{lemma:relation}
For any  formal difference operators $A,B\in \mathcal{F}[[\Lambda,\Lambda^{-1}]]$, we have
\begin{align}
&[A_{\Delta^*,\geq1},B]_{\Delta,\geq1}+[A,B_{\Delta,\geq1}]_{\Delta^*,\geq1}+[B_{\Delta,\geq1},A_{\Delta^*,\geq1}]=0.\label{commuAB}
\end{align}
\end{lem}
\begin{prf}
Firstly denote $``l.h.s."$ the left hand side of \eqref{commuAB}, by using Lemma \ref{lemma:relation1} and $[A_{\Delta^*,\geq1},B_{\Delta,\leq0}]_{\Lambda,\geq1}=[A_{\Delta^*,\leq0},B_{\Delta,\geq1}]_{\Lambda,<0}=0$, we have
\begin{align*}
&(l.h.s.)_{\Lambda,\geq1}=[A_{\Delta^*,\geq1},B_{\Delta,\geq1}]_{\Lambda,\geq1}+[B_{\Delta,\geq1},A_{\Delta^*,\geq1}]_{\Lambda,\geq1}=0,\\
&(l.h.s.)_{\Lambda,<0}=[A_{\Delta^*,\geq1},B_{\Delta,\geq1}]_{\Lambda,<0}+[B_{\Delta,\geq1},A_{\Delta^*,\geq1}]_{\Lambda,<0}=0.
\end{align*}
In addition, we can find that $(l.h.s.)_{\Lambda,[0]}=l.h.s.(1)=0$. Therefore, we conclude this lemma.
\end{prf}

The modified Toda hierarchy  is a family of evolution equations
\begin{eqnarray}
\frac{\pa L_a}{\pa x^{(1)}_k}=[(L_1^k )_{\Delta,\geq1 },L_a], \quad
\frac{\pa L_a}{\pa x^{(2)}_k}=[(L_2^k)_{\Delta^*,\geq1},L_a],\label{laxeqL1}
\end{eqnarray}
in infinite variables  $\bs{x}=(\bs{x}^{(1)},\bs{x}^{(2)})$ with $\bs{x}^{(a)}=
\big(x_1^{(a)},x_2^{(a)},\cdots\big)$ $(a=1,2)$, where $L_1, L_2$ are  pseudo difference operators of the forms,
\begin{align}
&L_1 = \sum_{m=-1}^\infty u^{(1)}_m\Lambda^{-m}\in \mathcal{F}((\Lambda^{-1})),\quad L_2 =\sum_{m= -1}^\infty u_m^{(2)}\Lambda^{m}\in\mathcal{F}((\Lambda)),\label{l1l2}
\end{align}
There exist two dressing operators of the form with $w^{(1)}_0\neq0$ and $w^{(2)}_0\neq0$
\begin{align}
&W_1=\sum_{i=0}^{\infty}w^{(1)}_i\Lambda^{-i}\in \mathcal{F}((\Lambda^{-1})),\quad W_2=\sum_{i=0}^{\infty}w^{(2)}_i\Lambda^{i}\in \mathcal{F}((\Lambda)),\label{w1w2form}
\end{align}
such that the Lax operators $L_1, L_2$ can be expressed in terms of the dressing operators as
\begin{align*}
  &L_1= W_1\Lambda W_1^{-1},\quad L_2=W_2\Lambda ^{-1}W_2^{-1}.
  \end{align*}
  The flows of modified Toda act on dressing operators by
\begin{eqnarray}
&&\frac{\pa W_1}{\pa x^{(1)}_k} =-(L_1^k )_{\Delta,\leq0 }W_1,\quad
  \frac{\pa W_2}{\pa x^{(1)}_k} = (L_1^k )_{\Delta,\geq1 }W_2,\nonumber\\
 &&\frac{\pa W_1}{\pa x^{(2)}_k} = (L_2^k)_{\Delta^*,\geq1}W_1,\quad
  \frac{\pa W_2}{\pa x^{(2)}_k} = -(L_2^k)_{\Delta^*,\leq0 }W_2.\label{evo}
\end{eqnarray}
The wave functions $\psi_a(s,\bs{x},z)$ and adjoint wave functions $\psi^*_a(s,\bs{x},z) (a=1,2)$ of  modified Toda  are defined by
 \begin{align}
  &\psi_1(s,\bs{x},z)
  =W_1e^{\xi(\bs{x}^{(1)},\Lambda)}(z^s),\quad \psi_2(s,\bs{x},z)
  =W_2e^{\xi(\bs{x}^{(2)},\Lambda^{-1})}(z^s),\label{wave}\\
    &\psi_1^*(s,\bs{x},z)
  =(W_1^{-1}\Lambda\Delta^{-1})^*e^{-\xi(\bs{x}^{(1)},\Lambda^{-1})}(z^{-s}),\label{adjointwave1}\\
  &\psi_2^*(s,\bs{x},z)
 =(W_2^{-1}\Delta^{*-1})^*e^{-\xi(\bs{x}^{(2)},\Lambda)}(z^{-s}),\label{adjointwave2}
\end{align}
satisfying   the following bilinear equation
 \begin{align}
&\oint_{C_{\infty}}\frac{dz}{2\pi \mathbf{i}z}\psi_1(s,\bs{x},z)\psi^*_1(s',\bs{x}^\prime,z)+\oint_{C_{0}}\frac{dz}{2\pi \mathbf{i}z}\psi_2(s,\bs{x},z)\psi^*_2(s',\bs{x}^\prime,z)=1,\label{waveequation1}
\end{align}
where   $C_{\infty}$ and $C_{0}$ mean the circle around $z=\infty$ and  $z=0$ respectively. Both $C_{\infty}$ and $C_{0}$ are anticlockwise.

 By construction, we have the auxiliary linear equations
\begin{align}
L_1(\psi_1(s,\bs{x},z))=z\psi_1(s,\bs{x},z),&\quad L_2(\psi_2(s,\bs{x},z))=z^{-1}\psi_2(s,\bs{x},z),\label{L*psi}\\
\frac{\pa \psi_a(s,\bs{x},z)}{\pa x^{(1)}_k} = (L_1^k )_{\Delta,\geq1 }(\psi_a(s,\bs{x},z)),&\quad
  \frac{\pa \psi_a(s,\bs{x},z)}{\pa x^{(2)}_k} = (L_2^k)_{\Delta^*,\geq1}(\psi_a(s,\bs{x},z)).\nonumber
\end{align}
According to \cite{mtoda}, there exist two tau functions $\tau_0$ and $\tau_1$ such that  wave functions has the following representations in terms of tau functions
\begin{align}
  &\psi_1(s,\bs{x},z)
  = \frac{G_1(z)\tau_0(s,\bs{x})}{\tau_1(s,\bs{x})}
    z^se^{\xi(\bs{x}^{(1)},z)},\quad \psi_2(s,\bs{x},z)
  = \frac{G_2(z)\Lambda\tau_0(s,\bs{x})}{\tau_1(s,\bs{x})}
    z^se^{\xi(\bs{x}^{(2)},z^{-1})},\label{wave}
\end{align}
where $[z^{-1}]=(z^{-1},z^{-2}/2,\cdots)$,  $\xi(\bs{x}^{(a)},z)=\sum\limits_{k=1}^{\infty}x^{(a)}_{k}z^k$ and $$G_1(z)=\exp(-\sum_{i=1}^{\infty}\frac{\pa_{x^{(1)}_i}}{iz^i}),\quad G_2(z)=\exp(-\sum_{i=1}^{\infty}\frac{\pa_{x^{(2)}_i}}{i}z^i).$$

For $k\in \Z_{+}$ and $a=1,2$, denote
\begin{align}
p^{(a)}_k=\pa_{x_k^{(a)}}, \quad p^{(a)}_{-k}=kx_{k}^{(a)}.
\end{align}
Introduce $p^{(a)}(\lambda)=\sum_{k\in\Z}\frac{p^{(a)}_k}{k\lambda^k}$, we can define the vertex operators for $a=1,2$ as follows,
\begin{align}
&X_a(\lambda,\mu)=:e^{p^{(a)}(\lambda)-p^{(a)}(\mu)}:=
e^{\xi(\bs{x}^{(a)},\mu)-\xi(\bs{x}^{(a)},\lambda)}e^{\xi(\tilde{\pa}_{\bs{x}^{(a)}},\lambda^{-1})-\xi(\tilde{\pa}_{\bs{x}^{(a)}},\mu^{-1})},\label{ver-oper}
\end{align}
where $\tilde{\pa}_{\bs{x}^{(a)}}=(\pa_{x_1^{(a)}},\frac{\pa_{x_2^{(a)}}}{2},\cdots)$, $:\cdot:$ stands for the normal-order product, which means that $p_k^{(a)}(k\geq0)$ must be placed
to the right of $p_k^{(a)}(k<0)$.

To prepare for the next section, we require the following proposition.
\begin{prp}
The wave and adjoint wave functions satisfy the following  bilinear equations in modified version
\begin{align}
&\oint_{C_{\infty}}\frac{dz}{2\pi \mathbf{i}z}\frac{1-z/\mu}{1-z/\lambda}\frac{X_1(\lambda,\mu)G_1(z)\tau_0(s,\bs{x})}{G_1(z)\tau_0(s,\bs{x})}\psi_1(s,\bs{x},z)\psi^*_1(s',\bs{x}^\prime,z)\nonumber\\
&+\oint_{C_{0}}\frac{dz}{2\pi \mathbf{i}z}\frac{X_1(\lambda,\mu)G_2(z)\tau_0(s+1,\bs{x})}{G_2(z)\tau_0(s+1,\bs{x})}\psi_2(s,\bs{x},z)\psi^*_2(s',\bs{x}^\prime,z)=\frac{X_1(\lambda,\mu)\tau_1(s,\bs{x})}{\tau_1(s,\bs{x})},\label{wavemd1}
\end{align}
and
\begin{align}
&\oint_{C_{0}}\frac{dz}{2\pi \mathbf{i}z}\frac{1-\frac{1}{z\mu}}{1-\frac{1}{z\lambda}}\frac{X_2(\lambda,\mu)G_2(z)\tau_0(s+1,\bs{x})}{G_2(z)\tau_0(s+1,\bs{x})}\psi_2(s,\bs{x},z)\psi^*_2(s',\bs{x}^\prime,z)\nonumber\\
&+\oint_{C_{\infty}}\frac{dz}{2\pi \mathbf{i}z}\frac{X_2(\lambda,\mu)G_1(z)\tau_0(s,\bs{x})}{G_1(z)\tau_0(s,\bs{x})}\psi_1(s,\bs{x},z)\psi^*_1(s',\bs{x}^\prime,z)=\frac{X_2(\lambda,\mu)\tau_1(s,\bs{x})}{\tau_1(s,\bs{x})}.\label{wavemd2}
\end{align}
\end{prp}
\begin{prf}
The proof is similar to that for the case of the Toda hierarchy\cite{ASVM}. In fact,  \eqref{wavemd1} (resp. \eqref{wavemd2}) is derived by shifting  $x^{(1)}$ (resp. $x^{(2)}$) into $x^{(1)}+[\lambda^{-1}]-[\mu^{-1}]$ (resp. $x^{(2)}+[\lambda^{-1}]-[\mu^{-1}]$) in \eqref{waveequation1}, and
multiplying both sides by $\frac{X_1(\lambda,\mu)\tau_1(s,\bs{x})}{\tau_1(s,\bs{x})}$ (resp. $\frac{X_2(\lambda,\mu)\tau_1(s,\bs{x})}{\tau_1(s,\bs{x})}$).
\end{prf}

\section{Additional symmetry of modified Toda hierarchy}
In this section, in order to construct additional symmetries of the modified Toda hierarchy, we first define the Orlov--Schulman operators by means of the dressing operators
\begin{align}
&M_1=S_1 s\Lambda^{-1}S^{-1}_1=W_1\big(s\Lambda^{-1}+\sum_{k=1}^{\infty}kx_k^{(1)}\Lambda^{k-1}\big)W^{-1}_1, \label{M1}\\ &M_2=-S_2s\Lambda S^{-1}_2=W_2\big(-s\Lambda+\sum_{k=1}^{\infty}kx_k^{(2)}\Lambda^{1-k}\big)W^{-1}_2\label{M2},
\end{align}
where $S_1=W_1e^{\xi(\bs{x}^{(1)},\Lambda)}$ and $S_2=W_2e^{\xi(\bs{x}^{(2)},\Lambda^{-1})}. $
Thus for any formal series, we have $f(M_1,L_1)=S_1f(s\Lambda^{-1},\Lambda)S_1^{-1}$ and $f(M_2,L_2)=S_2f(s\Lambda,\Lambda^{-1})S_2^{-1}$.

Direct calculations show that the operators $M_1, M_2$ satisfy
\begin{align}
[L_1,M_1]=&[L_2,M_2]=1,\label{L1M1}\\
M_1(\psi_1)=\pa_z(\psi_1),&\quad M_2(\psi_2)=\pa_{z^{-1}}(\psi_2),\nonumber
\end{align}
and the following evolution equations
\begin{eqnarray}
\frac{\pa M_a}{\pa x^{(1)}_k}=[(L_1^k )_{\Delta,\geq1 },M_a], \quad
\frac{\pa M_a}{\pa x^{(2)}_k}=[(L_2^k)_{\Delta^*,\geq1},M_a],\quad a=1,2,\quad k\geq1.\label{laxeqM1}
\end{eqnarray}

Introduce  the following  vector fields for $l\in \Z$ and $m\geq0$,
\begin{eqnarray}
&&\frac{\pa W_1}{\pa s^{(1)}_{m,l}} =-(M_1^mL_1^l )_{\Delta,\leq0 }W_1,\quad
  \frac{\pa W_2}{\pa s^{(1)}_{m,l}} = (M_1^mL_1^l )_{\Delta,\geq1 }W_2,\label{s1W}\\
 &&\frac{\pa W_1}{\pa s^{(2)}_{m,l}} = (M_2^mL_2^l )_{\Delta^*,\geq1}W_1,\quad
  \frac{\pa W_2}{\pa s^{(2)}_{m,l}} = -(M_2^mL_2^l )_{\Delta^*,\leq0 }W_2,\label{s2W}
\end{eqnarray}
which induce vector fields on the Lax operators $L_1$ and $L_2$
\begin{align}
&&\frac{\pa L_1}{\pa s^{(1)}_{m,l}}=-[(M_1^mL_1^l )_{\Delta,\leq0 },L_1], \quad
\frac{\pa L_1}{\pa s^{(2)}_{m,l}}=[(M_2^mL_2^l )_{\Delta^*,\geq1},L_1],\label{addL1}\\
&&\frac{\pa L_2}{\pa s^{(1)}_{m,l}}=[(M_1^mL_1^l )_{\Delta,\geq1 },L_2], \quad
\frac{\pa L_2}{\pa s^{(2)}_{m,l}}=-[(M_2^mL_2^l )_{\Delta^*,\leq0},L_2],\label{addL2}
\end{align}
where $s^{(a)}_{m,l}$ denotes additional variables.

According to Lemma \ref{lemma:relation}, it is easy to check that
\begin{prp}\label{pro:comm}
The vector fields  $\frac{\pa}{\pa{s^{(a)}_{m,l}}}$  commute with the time flows of   modified Toda. Namely, for any  $a,b=1,2$, one has
\begin{align}
\left[\frac{\pa}{\pa{s^{(a)}_{m,l}}},\frac{\pa}{\pa{x^{(b)}_{k}}}\right]=0, \quad m\geq0,\quad l\in\Z,\quad k\geq1.\label{add-time}
\end{align}
Moreover, the vector fields $\frac{\pa}{\pa{s^{(1)}_{m,l}}}, \frac{\pa}{\pa{s^{(2)}_{m,l}}}$ acting on the dressing operators $W_1,W_2$ (or on the
wave function $\psi_a(s,\bs{x},z)(a=1,2)$) satisfy
\begin{align}
&\left[\frac{\pa}{\pa{s^{(a)}_{m,l}}},\frac{\pa}{\pa{s^{(a)}_{k,n}}}\right]=\sum_{p,q}C^{p,q}_{m,l,k,n}\frac{\pa}{\pa{s^{(a)}_{p,q}}},\quad a=1,2,\label{s1s1}\\
&\left[\frac{\pa}{\pa{s^{(1)}_{m,l}}},\frac{\pa}{\pa{s^{(2)}_{k,n}}}\right]=0. \label{s1s2}
\end{align}
 Therefore   these additional symmetries acting on the wave function form a centerless $W_{\infty}\times W_{\infty}$ algebra, where  $C^{p,q}_{m,l,k,n}$ are structure constants of the algebra.
\end{prp}
\begin{prf}
Here we only prove \eqref{add-time} in the case $a=1$ and $b=2$, since other cases are completely similar. Indeed, by using \eqref{evo}, \eqref{laxeqM1} and \eqref{s1W}, it follows that
\begin{align*}
\left[\frac{\pa}{\pa{s^{(1)}_{m,l}}},\frac{\pa}{\pa{x^{(2)}_{k}}}\right]W_1=&\frac{\pa}{\pa{s^{(1)}_{m,l}}}\left((L_2^k)_{\Delta^*,\geq1}W_1\right) +\frac{\pa}{\pa{x^{(2)}_{k}}}\left((M_1^mL_1^l )_{\Delta,\leq0 }W_1\right)\\
=&\left[(M_1^mL_1^l )_{\Delta,\geq1 },L_2^k\right]_{\Delta^*,\geq1}W_1 -\left[(L_2^k)_{\Delta^*,\geq1},(M_1^mL_1^l )_{\Delta,\leq0 }\right]W_1\\
&+\left[(L_2^k)_{\Delta^*,\geq1}, M_1^mL_1^l\right] _{\Delta,\leq0 }W_1=0,
\end{align*}
where we have used Lemma \ref{lemma:relation} with $A=L_2^k$ and $B=M_1^mL_1^l$.
 \eqref{s1s1} and \eqref{s1s2}  can be proven by means of similar process in \cite{ASVM}.
\end{prf}

This  proposition implies the vector fields \eqref{addL1} and \eqref{addL2} give an additional symmetry for modified Toda hierarchy.
\begin{remark}
Note that the vector fields $\pa_{s^{(1)}_{0,l}}$ cannot be identified with $\pa_{x_l^{(1)}}$, though their actions on wave function coincide. Their actions on $M_1$ are not identical: $\frac{\pa M_1}{\pa s^{(1)}_{0,l}}=-[(L_1^l )_{\Delta,\leq0 },M_1]$ and $\frac{\pa M_1}{\pa x^{(1)}_l}=[(L_1^l )_{\Delta,\geq1 },M_1]$ and these are not the same. The same applies to $\pa_{s^{(2)}_{0,l}}$ and $\pa_{x_l^{(2)}}$
\end{remark}

Next we want to represent the additional symmetries \eqref{addL1} and \eqref{addL2} on the tau function of the modified Toda hierarchy. To this end we introduce   two generating functions of operators with parameters $\lambda$ and $\mu$
\begin{align}
&Y_a(\lambda,\mu)=\sum_{m=0}^{\infty}\frac{(\mu-\lambda)^m}{m!}\sum_{l=-\infty}^{\infty}\lambda^{-m-l-1}(M_a^mL_a^{m+l})=e^{(\mu-\lambda)M_a}\delta(\lambda,L_a),\quad a=1,2.\label{Y1Y2}
\end{align}
Here, the delta-function $\delta(\lambda,\mu)$ is defined as
 \begin{align}
 \label{delta}\delta(\lambda,\mu)=\frac{1}{\mu}\sum\limits_{i=-\infty}^{\infty}(\frac{\mu}{\lambda})^i=\frac{1}{\mu}\frac{1}{1-\lambda/\mu}+\frac{1}{\lambda}\frac{1}{1-\mu/\lambda},
 \end{align}
with the following property:
 $$\delta(\lambda,\mu)f(\mu)=\delta(\lambda,\mu)f(\lambda)$$
for any function $f(\mu)=\sum_{k\in\Z}f_k\mu^k$.
\begin{prp}\label{Y1wave}
The generators $Y_a(\lambda,\mu)$ acting on the wave functions can be represented as
\begin{align*}
&Y_1(\lambda,\mu)\big(\psi_1(s,\bs{x},z)\big)=\frac{\mathbb{X}_1(\lambda,\mu)G_1(z)\tau_0(s)}{G_1(z)\tau_0(s)}\psi_1(s,\bs{x},z)\delta(\lambda,z),\\
&Y_2(\lambda,\mu)\big(\psi_2(s,\bs{x},z)\big)=\frac{\mathbb{X}_2(\lambda,\mu)G_2(z)\Lambda\tau_0(s)}{G_2(z)\Lambda\tau_0(s)}\psi_2(s,\bs{x},z)\delta(\lambda,z^{-1}),
\end{align*}
where \begin{align}
\mathbb{X}_1(\lambda,\mu)=(\frac{\mu}{\lambda})^sX_1(\lambda,\mu),\quad \mathbb{X}_2(\lambda,\mu)=(\frac{\lambda}{\mu})^sX_2(\lambda,\mu).\label{X1-X2}
\end{align}
\end{prp}
\begin{prf}
Let us check the first equality first. In view of the definition of $Y_1(\lambda,\mu)$ and wave function $\psi_1(s,\bs{x},z)$,
\begin{align*}
Y_1(\lambda,\mu)\big(\psi_1(s,\bs{x},z)\big)&=S_1e^{(\mu-\lambda)s\Lambda^{-1}}\delta(\lambda,\Lambda)S^{-1}_1(\psi_1(s,\bs{x},z))\\
&=S_1e^{(\mu-\lambda)\pa_z}\delta(\lambda,z)(z^s)=\delta(\lambda,z)S_1(\mu^s)\\
&=\delta(\lambda,z) \frac{\tau_0(s,\bs{x}^{(1)}-[\Lambda^{-1}],\bs{x}^{(2)} )}{\tau_1(s,\bs{x}^{(1)},\bs{x}^{(2)})}
    e^{\xi(\bs{x}^{(1)},\Lambda)}(\mu^s)\\
    &=\delta(\lambda,z) \frac{\tau_0(s,\bs{x}^{(1)}-[\mu^{-1}]+[\lambda^{-1}]-[z^{-1}],\bs{x}^{(2)} )}{\tau_1(s,\bs{x}^{(1)},\bs{x}^{(2)})}
    e^{\xi(\bs{x}^{(1)},\mu)}\mu^s\\
    &=(\frac{\mu}{\lambda})^s\frac{X_1(\lambda,\mu)G_1(z)\tau_0(s)}{G_1(z)\tau_0(s)}\psi_1(s,\bs{x},z)\delta(\lambda,z).
\end{align*}
The verification of the second equality is similar. Thus ending the proof of proposition \ref{Y1wave}.
\end{prf}

Introduce the generators of the additional symmetries
 of modified Toda
\begin{align*}
&\pa^*_{\alpha_{k}}=\sum_{m=0}^{\infty}\frac{(\mu-\lambda)^m}{m!}\sum_{l=-\infty}^{\infty}\lambda^{-m-l-1}\pa_{s^{(k)}_{m,m+l}},\quad k=1,2,
\end{align*}
then by \eqref{s1W}, \eqref{s2W}, \eqref{Y1Y2}, it follows that
\begin{align*}
\pa^*_{\alpha_{1}}W_1 =-Y_1(\lambda,\mu)_{\Delta,\leq0 }W_1,&\quad
  \pa^*_{\alpha_{1}}W_2 = Y_1(\lambda,\mu)_{\Delta\geq1}W_2,\label{}\\
\pa^*_{\alpha_{2}}W_1 = Y_2(\lambda,\mu)_{\Delta^*,\geq1}W_1,&\quad
  \pa^*_{\alpha_{2}}W_2 = -Y_2(\lambda,\mu)_{\Delta^*,\leq0 }W_2.\label{}
\end{align*}
Next let us   consider the actions of $\pa^*_{\alpha_{k}}(k=1,2)$ on tau functions.
For this, let us see a lemma developed in \cite{ASVM}.
\begin{lem}\label{equationlemma}
Given two operators $U$ and $V$, then
\begin{align*}
&U(s,\Lambda)V(s,\Lambda)={\rm Res}_{z}z^{-1}\sum_{j\in\mathbb{Z}}U(s,\Lambda)(z^s)\cdot V^*(s+j,\Lambda)(z^{-s-j})\Lambda^j,
\end{align*}
where ${\rm Res}_\lambda\sum_ia_i\lambda^i=a_{-1}$.
\end{lem}

Denote
\begin{align*}
f_1(z)=-G_1(z)\left(\frac{\mathbb{X}_1(\lambda,\mu)\tau_0(s,\bs{x})}{\tau_0(s,\bs{x})}\right),&\quad f_2(z)=-G_2(z)\left(\frac{\mathbb{X}_2(\lambda,\mu)\Lambda(\tau_0(s,\bs{x}))}{\Lambda(\tau_0(s,\bs{x}))}\right),\\
g_1(z)=G_1(z)\left(\frac{\mathbb{X}_2(\lambda,\mu)\tau_0(s,\bs{x})}{\tau_0(s,\bs{x})}\right),&\quad g_2(z)=G_2(z)\left(\frac{\mathbb{X}_1(\lambda,\mu)\Lambda(\tau_0(s,\bs{x}))}{\Lambda(\tau_0(s,\bs{x}))}\right),
\end{align*}
then it follows from \eqref{L*psi}, \eqref{delta}, Proposition \ref{Y1wave}, and the following identities
\begin{align*}
&G_1(z)X_1(\lambda,\mu)=X_1(\lambda,\mu)G_1(z)\frac{1-\mu/z}{1-\lambda/z},\quad
&G_2(z)X_2(\lambda,\mu)=X_2(\lambda,\mu)G_2(z)\frac{1-z\mu}{1-z\lambda},
\end{align*}
we can then obtain the relations
\begin{align*}
&\big(Y_1(\lambda,\mu)(\mu-\lambda)-f_1(L_1)\big)\big(\psi_1(s,\bs{x},z)\big)=\frac{\mathbb{X}_1(\lambda,\mu)G_1(z)\tau_0(s)}{G_1(z)\tau_0(s)}\frac{\mu}{\lambda}\frac{1-z/\mu }{1- z/\lambda}\psi_1,\\
&\big(Y_2(\lambda,\mu)(\mu-\lambda)-f_2(L_2^{-1})\big)\big(\psi_2(s,\bs{x},z)\big)=\frac{\mathbb{X}_2(\lambda,\mu)G_2(z)\Lambda\tau_0(s)}{G_2(z)\Lambda\tau_0(s)}\frac{\mu}{\lambda}\frac{1-1/z\mu}{1-1/z\lambda}\psi_2.
\end{align*}
Hence by using the bilinear equations \eqref{wavemd1} and \eqref{wavemd2}, we have
\begin{align}
&\oint_{C_{\infty}}\frac{dz}{2\pi \mathbf{i}z}\frac{\lambda}{\mu}\big(Y_1(\lambda,\mu)(\mu-\lambda)-f_1(L_1)\big)\big(\psi_1(s,\bs{x},z)\big)\psi^*_1(s',\bs{x}^\prime,z)\nonumber\\
&+\oint_{C_{0}}\frac{dz}{2\pi \mathbf{i}z}g_2(L_2^{-1})(\psi_2(s,\bs{x},z))\psi^*_2(s',\bs{x}^\prime,z)=\frac{\mathbb{X}_1(\lambda,\mu)\tau_1(s,\bs{x})}{\tau_1(s,\bs{x})},\label{Y1f1g2}\\
&\oint_{C_{0}}\frac{dz}{2\pi \mathbf{i}z}\frac{\lambda}{\mu}\big(Y_2(\lambda,\mu)(\mu-\lambda)-f_2(L_2^{-1})\big)(\psi_2(s,\bs{x},z))\psi^*_2(s',\bs{x}^\prime,z)\nonumber\\
&+\oint_{C_{\infty}}\frac{dz}{2\pi \mathbf{i}z}g_1(L_1)(\psi_1(s,\bs{x},z))\psi^*_1(s',\bs{x}^\prime,z)=\frac{\mathbb{X}_2(\lambda,\mu)\tau_1(s,\bs{x})}{\tau_1(s,\bs{x})}.\label{Y2f2g1}
\end{align}
\begin{prp}\label{pro:Y1Y2}

\begin{align}
&Y_1(\lambda,\mu)_{\Delta\geq1}=\frac{\mu}{\lambda(\mu-\lambda)}\big(g_2(L_2^{-1})-\frac{\mathbb{X}_1(\lambda,\mu)\tau_1(s,\bs{x})}{\tau_1(s,\bs{x})}\big),\label{Y1delta>0}\\
&Y_1(\lambda,\mu)_{\Delta\leq0}=\frac{1}{\mu-\lambda}\big(f_1(L_1)+\frac{\mu}{\lambda}\frac{\mathbb{X}_1(\lambda,\mu)\tau_1(s,\bs{x})}{\tau_1(s,\bs{x})}\big),\label{Y1delta<1}
\end{align}
and
\begin{align}
&Y_2(\lambda,\mu)_{\Delta^*\geq1}=\frac{\mu}{\lambda(\mu-\lambda)}\big(g_1(L_1)-\frac{\mathbb{X}_2(\lambda,\mu)\tau_1(s,\bs{x})}{\tau_1(s,\bs{x})}\big),\label{Y2delta*>0}\\
&Y_2(\lambda,\mu)_{\Delta^*\leq0}=\frac{1}{\mu-\lambda}\big(f_2(L_2^{-1})+\frac{\mu}{\lambda}\frac{\mathbb{X}_2(\lambda,\mu)\tau_1(s,\bs{x})}{\tau_1(s,\bs{x})}\big).\label{Y2delta*<1}
\end{align}
\end{prp}
\begin{prf}
Firstly, by shifting $s'$ to $s+i$, summing over $i$, and setting $x^{(a)'} = x^{(a)}$ for $a = 1, 2$ in \eqref{Y1f1g2}, we can apply Lemma \ref{equationlemma} to obtain:
\begin{align*}
\frac{\lambda}{\mu}\Big(Y_1(\lambda,\mu)(\mu-\lambda)-f_1(L_1)\Big)\Lambda\Delta^{-1}+g_2(L_2^{-1})\Delta^{*-1}=\sum_{i\in\Z}\frac{\mathbb{X}_1(\lambda,\mu)\tau_1(s,\bs{x})}{\tau_1(s,\bs{x})}\Lambda^i,
\end{align*}
This leads to equations \eqref{Y1delta>0} and \eqref{Y1delta<1} by taking  $(\cdot)_{\Lambda\geq1}$ and $(\cdot)_{\Lambda\leq0}$, respectively. Similarly, equations \eqref{Y2delta*>0} and \eqref{Y2delta*<1} can be derived from \eqref{Y2f2g1} by similar steps.
\end{prf}
\begin{thm}\label{thm:ASvM}
$\mathbf{(ASvM\ formula)}$\ The vector fields of types $\pa^*_{\alpha_{k}}(k=1,2)$ on the wave function and the vertex operators  $\mathbb{X}_a(\lambda,\mu)(a=1,2)$ on tau functions are related as follows
\begin{align*}
&\frac{\pa^*_{\alpha_{1}}\psi_1}{\psi_1}=\frac{1}{\mu-\lambda}\Big(G_1(z)\big(\frac{\mathbb{X}_1(\lambda,\mu)\tau_0(s,\bs{x})}{\tau_0(s,\bs{x})}\big)-\frac{\mu}{\lambda}\frac{\mathbb{X}_1(\lambda,\mu)\tau_1(s,\bs{x})}{\tau_1(s,\bs{x})}\Big),\\
&\frac{\pa^*_{\alpha_{1}}\psi_2}{\psi_2}=\frac{1}{\mu-\lambda}\frac{\mu}{\lambda}\Big(G_2(z)\big(\frac{\mathbb{X}_1(\lambda,\mu)\tau_0(s+1,\bs{x})}{\tau_0(s+1,\bs{x})}\big)-\frac{\mathbb{X}_1(\lambda,\mu)\tau_1(s,\bs{x})}{\tau_1(s,\bs{x})}\Big),
\end{align*}
and
\begin{align*}
&\frac{\pa^*_{\alpha_{2}}\psi_1}{\psi_1}=\frac{1}{\mu-\lambda}\frac{\mu}{\lambda}\Big(G_1(z)\big(\frac{\mathbb{X}_2(\lambda,\mu)\tau_0(s,\bs{x})}{\tau_0(s,\bs{x})}\big)-\frac{\mathbb{X}_2(\lambda,\mu)\tau_1(s,\bs{x})}{\tau_1(s,\bs{x})}\Big),\\
&\frac{\pa^*_{\alpha_{2}}\psi_2}{\psi_2}=\frac{1}{\mu-\lambda}\Big(G_2(z)\big(\frac{\mathbb{X}_2(\lambda,\mu)\tau_0(s+1,\bs{x})}{\tau_0(s+1,\bs{x})}\big)-\frac{\mu}{\lambda}\frac{\mathbb{X}_2(\lambda,\mu)\tau_1(s,\bs{x})}{\tau_1(s,\bs{x})}\Big).
\end{align*}
\end{thm}
At last according to the relations  between the wave functions $\psi_a(a=1,2)$ and the tau functions $\tau_a(a=1,2)$ shown in \eqref{wave}, we get
the following proposition.
\begin{prp}\label{pro:paalpha}
The actions of vector fields of types $\pa^*_{\alpha_{a}}(a=1,2)$ on the tau functions are given as follows
\begin{align*}
&\pa^*_{\alpha_{1}}\tau_0(s,\bs{x})=\frac{1}{\mu-\lambda}\mathbb{X}_1(\lambda,\mu)\tau_0(s,\bs{x}),\quad
\pa^*_{\alpha_{1}}\tau_1(s,\bs{x})=\frac{\mu}{\lambda}\frac{1}{\mu-\lambda}\mathbb{X}_1(\lambda,\mu)\tau_1(s,\bs{x}),
\end{align*}
and
\begin{align*}
&\pa^*_{\alpha_{2}}\tau_0(s,\bs{x})=\frac{\mu}{\lambda}\frac{1}{\mu-\lambda}\mathbb{X}_2(\lambda,\mu)\tau_0(s,\bs{x}),\quad
\pa^*_{\alpha_{2}}\tau_1(s,\bs{x})=\frac{\mu}{\lambda}\frac{1}{\mu-\lambda}\mathbb{X}_2(\lambda,\mu)\tau_1(s,\bs{x}).
\end{align*}
\end{prp}

We want to represent the additional symmetries \eqref{s1W} and \eqref{s2W} with the tau functions. To this end we expand the vertex operators in \eqref{X1-X2} formally as
\begin{align}
&\mathbb{X}_a(\lambda,\mu)=\sum_{m=0}^{\infty}\frac{(\mu-\lambda)^m}{m!}\sum_{l=-\infty}^{\infty}\lambda^{-m-l}W_{s,l,m}^{(a)},\quad a=1,2
\end{align}
where
\begin{align}
&W_{s,l,m}^{(1)}=\sum_{i=0}^m\frac{m!}{(m-i)!}\binom{s}{i}\res_{\lambda}\lambda^{m+l-i}\pa^{m-i}_{\mu}|_{\mu=\lambda}X_1(\lambda,\mu),\\
&W_{s,l,m}^{(2)}=\sum_{i=0}^m\frac{m!}{(m-i)!}\binom{-s}{i}\res_{\lambda}\lambda^{m+l-i}\pa^{m-i}_{\mu}|_{\mu=\lambda}X_2(\lambda,\mu).
\end{align}
It is easy to calculate
\begin{align*}
&W_{s,l,0}^{(1)}=W_{s,l,0}^{(2)}=\delta_{l,0},\quad W_{s,l,1}^{(1)}=p^{(1)}_{l}+s\delta_{l,0},\quad W_{s,l,1}^{(2)}=p^{(2)}_{l}-s\delta_{l,0},\\ &W_{s,l,2}^{(1)}=\sum_{i+j=l}:p^{(1)}_ip^{(1)}_j:-(l-2s+1)p^{(1)}_l+\frac{s(s-1)}{2}\delta_{l,0},\\
&W_{s,l,2}^{(2)}=\sum_{i+j=l}:p^{(2)}_ip^{(2)}_j:-(l+2s+1)p^{(2)}_l+\frac{s(s+1)}{2}\delta_{l,0}.
\end{align*}

\begin{cor}\label{cor:addtau}
The actions of additional symmetries  $\pa_{s^{(1)}_{m,l}}$ and $\pa_{s^{(2)}_{m,l}}$ on the tau functions are given as follows
\begin{align*}
&\pa_{s^{(1)}_{m,m+l}}\tau_0(s,\bs{x})=\frac{W_{s,l,m+1}^{(1)}}{m+1}\tau_0(s,\bs{x}),\quad
\pa_{s^{(1)}_{m,m+l}}\tau_1(s,\bs{x})=\frac{W_{s+1,l,m+1}^{(1)}}{m+1}\tau_1(s,\bs{x}),
\end{align*}
and
\begin{align*}
&\pa_{s^{(2)}_{m,m+l}}\tau_0(s,\bs{x})=\frac{W_{s-1,l,m+1}^{(2)}}{m+1}\tau_0(s,\bs{x}),\quad
\pa_{s^{(2)}_{m,m+l}}\tau_1(s,\bs{x})=\frac{W_{s-1,l,m+1}^{(2)}}{m+1}\tau_1(s,\bs{x}).
\end{align*}
\end{cor}
\section{Bigraded Toda hierarchy and virasoro symmetry}
Given two arbitrary integers $N$ and $M$, the $(N,M)$--bigraded modified Toda is derived by imposing the constraint $L_1^N=L_2^M$, $ (L_1^N+L_2^M)(1)=0$ on modified Toda Lax operators $(L_1,L_2)$ shown in \eqref{l1l2}, Denote $\mathcal{L}=L_1^N=L_2^M$ by  the Lax operator of  $(N,M)$--BMTH,  then it has the following form
\begin{align*}
&\mathcal{L}:=a_{N}\Lambda^{N}+a_{N-1}\Lambda^{N-1}+\cdots +a_{-M} \Lambda^{-M},\quad a_{N}\neq 0,\quad a_{-M}\neq 0\\
&\sum\limits_{i=-M}^{N}a_i=0,
\end{align*}
satisfying the following Lax equations,
\begin{eqnarray}
&&\frac{\pa \mathcal{L}}{\pa x^{(1)}_k}=\big[(L_1^{k} )_{\Delta,\geq1 },\mathcal{L}\big],\quad\frac{\pa \mathcal{L}}{\pa x^{(2)}_k}=\big[(L_2^{k})_{\Delta^*,\geq1},\mathcal{L}\big],\label{laxequation}
\end{eqnarray}
Notice that \eqref{laxequation} can also be written as \eqref{evo}  via the dressing operators.

For $l\geq-1$, introduce the following vector field $\pa_{s_l}$ as the linear combination
\begin{align}
\pa_{s_l}:=\frac{1}{N}\pa _{s^{(1)}_{1,Nl+1}}+\frac{1}{M}\pa _{s^{(2)}_{1,Ml+1}}\label{pasl}
\end{align}
of the additional symmetries \eqref{s1W} and \eqref{s2W} for the modified Toda hierarchy.
Namely,
\begin{eqnarray}
&&\frac{\pa W_1}{\pa s_l} =\Big(\frac{1}{M}(M_2L_2^{Ml+1} )_{\Delta^*,\geq1}-\frac{1}{N}(M_1L_1^{Nl+1} )_{\Delta,\leq0 } \Big)W_1,\label{W1sl}\\
  &&\frac{\pa W_2}{\pa s_l} = \Big(\frac{1}{N}(M_1L_1^{Nl+1} )_{\Delta,\geq1 }-\frac{1}{M}(M_2L_2^{Ml+1})_{\Delta^*,\leq0 }\Big)W_2. \label{W2sl}
\end{eqnarray}
It follows from \eqref{L1M1} and $L_1^N=L_2^M$ that  for $l\geq-1$
\begin{eqnarray*}
\Big[\frac{1}{N}M_1L_1^{Nl+1}-\frac{1}{M}M_2L_2^{Ml+1}, \mathcal{L}\Big] =0.
\end{eqnarray*}
Since the operator $\mathcal{L}$ can be written as $\mathcal{L}=W_1\Lambda^{N}W^{-1}_1=W_2\Lambda^{-M}W^{-1}_2$, then the equations \eqref{W1sl} and \eqref{W2sl} can be redefined by
\begin{align}
\frac{\pa \mathcal{L}}{\pa s_l} &=\Big[\frac{1}{M}(M_2L_2^{Ml+1} )_{\Delta^*,\geq1}-\frac{1}{N}(M_1L_1^{Nl+1} )_{\Delta,\leq0 },\mathcal{L}\Big],\nonumber\\
 &= \Big[\frac{1}{N}(M_1L_1^{Nl+1} )_{\Delta,\geq1 }-\frac{1}{M}(M_2L_2^{Ml+1})_{\Delta^*,\leq0 },\mathcal{L}\Big].\label{Lsl}
\end{align}
\begin{lem}
The Lax operator $\mathcal{L}$ remains the constraint condition $\mathcal{L}(1)=0$ under the vector fields $\pa_{s_l}(l\geq-1)$.
\end{lem}
\begin{proof}
Indeed, the vector field $\pa_{s_l}$ acting on $\mathcal{L}$ reads
\begin{align*}
\frac{\pa \mathcal{L}}{\pa s_l} &=\mathcal{L}^{l+1}+\Big[\frac{1}{M}(M_2L_2^{Ml+1} )_{\Delta^*,\geq1}+\frac{1}{N}(M_1L_1^{Nl+1} )_{\Delta,\geq1 },\mathcal{L}\Big],
\end{align*}
which implies $\frac{\pa \mathcal{L}}{\pa s_l}(1)=0$ by using  $\mathcal{L}(1)=0$ and $A_{\Delta^*,\geq1}(1)=A_{\Delta,\geq1}(1)=0$ for arbitrary $A\in\mathcal{F}[[\Lambda,\Lambda^{-1}]] $.
\end{proof}
\begin{prp}\label{pro:bmth}
The vector fields $\pa_{s_l}(l\geq-1)$ give symmetries  of the $(N,M)$--BMTH hierarchy, namely,
for any  $a=1,2$, $k\geq1$ and $l\geq-1$, one has
\begin{align*}
\left[\pa_{s_{l}},\pa_{x^{(a)}_{k}}\right]=0.
\end{align*}
Further, the vector fields satisfy the virasoro commutation relations
\begin{align*}
\left[\pa_{s_m},\pa_{s_n}\right]=(n-m)\pa_{s_{m+n}}.
\end{align*}
\end{prp}
\begin{proof}
Firstly the first assertion is a consequence of Proposition \ref{pro:comm}. For the second assertion, it follows from \eqref{s1s1}, \eqref{s1s2} and \eqref{pasl},
\begin{align*}
&\left[\pa_{s_m},\pa_{s_n}\right]=\left[\frac{1}{N}\pa _{s^{(1)}_{1,Nm+1}}+\frac{1}{M}\pa _{s^{(2)}_{1,Mm+1}},\frac{1}{N}\pa _{s^{(1)}_{1,Nn+1}}+\frac{1}{M}\pa _{s^{(2)}_{1,Mn+1}}\right]\\
&=\frac{1}{N}(n-m)\pa _{s^{(1)}_{1,N(m+n)+1}}+ \frac{1}{M}(n-m)\pa _{s^{(2)}_{1,M(m+n)+1}}\\
&=(n-m)\pa_{s_{m+n}}.
\end{align*}
where we have used $C^{p,q}_{1,m+1,1,n+1}=(n-m)\delta_{p,m+n-1}\delta_{q,1}.$
\end{proof}
\noindent {\bf Example:}
Here are some examples of the additional symmetries of the $(1,1)$--BMTH hierarchy, that is
\begin{align*}
\mathcal{L}:=a_{1}\Lambda+a_{0} +a_{-1} \Lambda^{-1},\quad \text{with $a_{-1}+a_0+a_1=0$}.
\end{align*}
Note that by Lemma \ref{lemma:relation}, we can find that $(W_2s\Lambda W^{-1}_2)_{\Delta^*,\geq1}=(W_2s W^{-1}_2)_{\Delta^*,\geq1}=0$, $(W_1s\Lambda^{-1}W^{-1}_1)_{\Delta,\leq0 }=W_1s\Lambda^{-1}W^{-1}_1$ and $(W_1sW^{-1}_1)_{\Delta,\leq0 }=W_1sW^{-1}_1$, then it follows from \eqref{M1},  \eqref{M2} and
    \eqref{Lsl} that
  \begin{align*}
&\frac{\pa \mathcal{L}}{\pa s_{-1}}
=1+\sum_{k=1}^{\infty}kx_k^{(1)}\frac{\pa \mathcal{L}}{\pa x^{(1)}_{k-1}}+\sum_{k=1}^{\infty}kx_k^{(2)}\frac{\pa \mathcal{L}}{\pa x^{(2)}_{k-1}},\\
&\frac{\pa \mathcal{L}}{\pa s_{0}}=
\mathcal{L}+\sum_{k=1}^{\infty}kx_k^{(1)}\frac{\pa \mathcal{L}}{\pa x^{(1)}_k}+\sum_{k=1}^{\infty}kx_k^{(2)}\frac{\pa \mathcal{L}}{\pa x^{(2)}_k}.
\end{align*}
  which leads to
  \begin{align*}
&\frac{\pa a_i}{\pa s_{-1}}=\sum_{k=1}^{\infty}kx_k^{(1)}\frac{\pa a_i}{\pa x^{(1)}_{k-1}}+\sum_{k=1}^{\infty}kx_k^{(2)}\frac{\pa a_i}{\pa x^{(2)}_{k-1}}+\delta_{i,0},\quad i=-1,0,1.\\
&\frac{\pa a_i}{\pa s_{-1}}=\sum_{k=1}^{\infty}kx_k^{(1)}\frac{\pa a_i}{\pa x^{(1)}_{k}}+\sum_{k=1}^{\infty}kx_k^{(2)}\frac{\pa a_i}{\pa x^{(2)}_{k}}+a_i,\quad i=-1,0,1.
\end{align*}
\begin{cor}\label{cor:tau}
With tau functions $\tau_0, \tau_1$ reduced from that of the modified Toda hierarchy, the
 actions of above symmetries on the tau functions can be written as
 \begin{align*}
 &\pa_{s_l}\tau_0(s,\bs{x})=\frac{1}{2}\Big(\frac{1}{N}W_{s,Nl,2}^{(1)}+\frac{1}{M}W_{s-1,Ml,2}^{(2)}\Big)\tau_0(s,\bs{x}),\\
  &\pa_{s_l}\tau_1(s,\bs{x})=\frac{1}{2}\Big(\frac{1}{N}W_{s+1,Nl,2}^{(1)}+\frac{1}{M}W_{s-1,Ml,2}^{(2)}\Big)\tau_1(s,\bs{x}).
 \end{align*}
\end{cor}
\section{Conclusions and Discussions}
In this section, we summarize the results presented above.  we have successfully  constructed the additional symmetries $\pa_{s^{(a)}_{m,l}}$ in \eqref{s1W} and \eqref{s2W} for the modified Toda hierarchy, where  $a=1,2, m\geq0$ and $l\in\Z$. Furthermore it has been shown  that these additional flows form a centerless $W_{\infty}\times W_{\infty}$ algebra in Proposition \ref{pro:comm}. Building on this, we also derive the corresponding AsvM formula in Theorem \ref{thm:ASvM}, which links the additional symmetries on the wave functions and the actions of vertex operators  $\mathbb{X}_a(\lambda,\mu)(a=1,2)$ in \eqref{X1-X2} on tau functions $\tau_0$ and $\tau_1$. It should be noted that the existence of two tau functions brings
much difficulty compared with the Toda case, which has only a single tau function.
The Virasoro additional symmetries of the  $(N,M)$--BMTH are defined in \eqref{W1sl} and \eqref{W2sl}, whose commutation relations are presented in Proposition \ref{pro:bmth} and its actions on the tau functions  are showed in Corollary \ref{cor:tau}.
 Notably, the additional symmetries are pivotal  not only in string theory but also in the construction of constraint-type \cite{Cheng-cKP} and extended \cite{Liu-extended} integrable systems.
  We expect that
  the results presented in this paper will be helpful for further studies of the modified Toda system, as well as for the construction of other types of integrable systems related to the modified Toda hierarchy.


\begin{thebibliography}{99}
\bibitem{ASVM}M. Adler, T. Shiota and P. van Moerbeke,
A Lax representation for the vertex operator and the central extension.
Comm. Math. Phys. 171 (1995)  547--588.




\bibitem{CLT-mKP}J. P. Cheng, M. H. Li and K. L. Tian,
On the modified KP hierarchy: tau functions, squared eigenfunction symmetries and additional symmetries. J. Geom. Phys. 134 (2018) 19--37.
\bibitem{Cheng-cKP}
Y. Cheng, Constraints of the Kadomtsev-Petviashvili hierarchy. J. Math. Phys. 33 (1992) 3774--3782.

\bibitem{DKJM}
   E. Date, M.  Kashiwara, M. Jimbo and T. Miwa, Transformation groups for soliton
equations. Nonlinear integrable systems¡ªclassical theory and quantum theory, 39--119, World Sci. Publishing, Singapore, 1983.


\bibitem{Dick-string}
L. A. Dickey, Additional symmetries of KP, Grassmannian, and the string equation. Modern Phys. Lett. A. 8 (1993)  1259--1272.

\bibitem{add-KP}
L. A. Dickey, On additional symmetries of the KP hierarchy and Sato¡¯s B\"acklund
transformation. Comm. Math. Phys. 167 (1995) 227--233.



\bibitem{HTFM-BKP}  J. S. He,  K. L. Tian, A.  Foerster and W. X. Ma, Additional symmetries and string equation of the CKP hierarchy. Lett. Math. Phys. 81 (2007) 119--134.


\bibitem{Jimbo1983}M. Jimbo, T. Miwa, Solitons and infinite-dimensional Lie algebras, Publ. Res. Inst. Math. Sci. 19 (1983) 943--1001.


\bibitem{Krichever2022}
I. Krichever and A. Zabrodin, Constrained Toda hierarchy and turning points of the Ruijsenaars-Schneider model. Lett. Math. Phys. 112 (2022) 23.

\bibitem{Krichever2023}
I. Krichever and A. Zabrodin, Toda lattice with constraint of type B. Phys. D. 453 (2023)  133827.
\bibitem{Liu-extended}
X. J. Liu, Y. B. Zeng and R. L. Lin,  A new extended KP hierarchy. Phys. Lett. A. 372 (2008)  3819--3823.

\bibitem{OS}  A. Y.  Orlov and E. I. Schulman, Additional symmetries for integrable equations and conformal algebra representation. Lett. Math. Phys. 12 (1986)  171--179.


\bibitem{Prokofev2023}
 V. V. Prokofev and  A. V. Zabrodin, Tau-function of the B-Toda hierarchy.  Theoret. Math. Phys. 217 (2023) 1673--1688.


\bibitem{mtoda}
W. J. Rui, W. C. Guan, Y. Yang and J. P. Cheng, The modified Toda hierarchy. arXiv:2408.09450.


\bibitem{Tu-BKP}  M. H.  Tu,
On the BKP hierarchy: additional symmetries, Fay identity and Adler-Shiota-van Moerbeke formula. Lett. Math. Phys. 81 (2007)  93--105.


\bibitem{van-BKP} J. van de Leur, The Adler¨CShiota¨Cvan Moerbeke formula for the BKP hierarchy. J. Math. Phys. 36 (1995) 4940--4951.

\bibitem{van2015}
J. van de Leur and A. Y. Orlov, Pfaffian and determinantal tau functions. Lett. Math. Phys. 105 (2015) 1499--1531.



\bibitem{Wu-2BKP} C. Z.  Wu,
From additional symmetries to linearization of Virasoro symmetries. Phys. D.  249 (2013) 25--37.

\bibitem{BMTH}Y. Yang, W. J. Rui and  J. P. Cheng, Bigraded modified Toda hierarchy and its extensions, Phys. D. 469 (2024) 134343.

\end{thebibliography}
\end{document}